\documentclass[11pt,reqno,letterpaper]{amsart}

\usepackage[margin=1in]{geometry}
\usepackage[T1]{fontenc}
\usepackage{lmodern}
\usepackage{microtype}
\usepackage{amsmath,amssymb,mathtools}
\usepackage{xcolor}
\usepackage{aliascnt}
\usepackage[colorlinks=true,allcolors=blue,hypertexnames=false]{hyperref}
\usepackage[noabbrev,capitalize,nameinlink]{cleveref}

\newtheorem{theorem}{Theorem}[section]

\newaliascnt{proposition}{theorem}
\newtheorem{proposition}[proposition]{Proposition}
\aliascntresetthe{proposition}

\newaliascnt{lemma}{theorem}

\aliascntresetthe{lemma}

\newaliascnt{corollary}{theorem}
\newtheorem{corollary}[corollary]{Corollary}
\aliascntresetthe{corollary}

\newaliascnt{remark}{theorem}

\aliascntresetthe{remark}

\newaliascnt{definition}{theorem}

\aliascntresetthe{definition}

\crefname{theorem}{theorem}{theorems}
\Crefname{theorem}{Theorem}{Theorems}

\crefname{proposition}{proposition}{propositions}
\Crefname{proposition}{Proposition}{Propositions}

\crefname{lemma}{lemma}{lemmas}
\Crefname{lemma}{Lemma}{Lemmas}

\crefname{corollary}{corollary}{corollaries}
\Crefname{corollary}{Corollary}{Corollaries}

\crefname{remark}{remark}{remarks}
\Crefname{remark}{Remark}{Remarks}

\crefname{definition}{definition}{definitions}
\Crefname{definition}{Definition}{Definitions}

\newcommand{\R}{\mathbb R}
\newcommand{\E}{\mathbb E}
\newcommand{\Prob}{\mathbb P}
\newcommand{\Sym}{\operatorname{Sym}}
\newcommand{\rank}{\operatorname{rank}}
\newcommand{\tr}{\operatorname{tr}}
\newcommand{\spann}{\operatorname{span}}

\newcommand{\norm}[1]{\left\lVert #1\right\rVert_2}
\newcommand{\abs}[1]{\left\lvert #1\right\rvert}
\newcommand{\ip}[2]{\left\langle #1,#2\right\rangle}
\newcommand{\Dfun}[2]{D(#1,#2)}

\title[The sharp dimension bound in the JL lemma]
{The Sharp Dimension Bound in the Johnson--Lindenstrauss Lemma}

\author{Vishesh Jain}
\address{Department of Mathematics, Statistics, and Computer Science,
University of Illinois Chicago, Chicago, IL 60607, USA}
\email{visheshj@uic.edu}

\begin{document}

\begin{abstract}
The Johnson--Lindenstrauss lemma asserts that every set of $n$ points in $d$-dimensional Euclidean space embeds into $O(\varepsilon^{-2}\log n)$-dimensional Euclidean space with distortion at most $1+\varepsilon$. Larsen and Nelson conjectured that the optimal target dimension throughout the full range of the parameters $n,d, \varepsilon$ is 
\[
\Theta\left(\min\left\{d,n-1,\frac{\log(2+\varepsilon^2n)}{\varepsilon^2}\right\}\right).
\]
We resolve this conjecture in the affirmative. In fact, we prove the stronger statement that the upper bound is attained by a linear map. The matching lower bound, due to Larsen--Nelson and Alon--Klartag, holds even for nonlinear embeddings.
\end{abstract}

\maketitle

\section{Introduction}\label{sec:introduction}

A fundamental problem in the theory of metric embeddings is to determine,
for given $n,d\geq 2$ and $0<\varepsilon<1/2$, the smallest integer $r$
such that every $n$-point set $X\subseteq\R^d$ admits an embedding
$f:X\to(\R^r,\|\cdot\|_2)$ with distortion at most $1+\varepsilon$. By
this we mean that there exists $\lambda>0$ such that
\[
 \lambda\norm{x-y}
 \leq \norm{f(x)-f(y)}
 \leq (1+\varepsilon)\lambda\norm{x-y}
\]
for all $x,y\in X$.

The celebrated Johnson--Lindenstrauss lemma \cite{JL} asserts that one may
take
\[
 r=O(\varepsilon^{-2}\log n).
\]
Moreover, the embedding may be chosen as the restriction of a random linear map
$\R^d\to\R^r$ whose distribution is independent of $X$; for example, a suitably
normalized Gaussian matrix succeeds with positive probability. A second 
upper bound follows from the affine span of $X$: after translating the
point set, one may realize it isometrically in dimension at most
$\min\{d,n-1\}$. Taking the better of these two alternatives gives the
previously best known upper bound
\[
 r=O\left(
 \min\left\{
 d,n-1,\varepsilon^{-2}\log n
 \right\}
 \right).
\]

Determining whether this upper bound is optimal led to a sequence of lower
bounds. Johnson and Lindenstrauss already showed that, for fixed distortion,
the logarithmic dependence on $n$ cannot be removed. Alon \cite{Alon} proved that
there are $n$-point subsets of $\R^n$ for which every embedding with
distortion at most $1+\varepsilon$ requires
\[
 \Omega\left(
 \min\left\{
 n,\frac{\varepsilon^{-2}\log n}{\log(1/\varepsilon)}
 \right\}
 \right)
\]
dimensions; the same asymptotic lower bound can also be
deduced from earlier spherical-code bounds of Welch and Levenshtein
\cite{Welch,Levenshtein}. 

Alon's lower bound
falls short of the Johnson--Lindenstrauss upper bound by a factor of
order $\log(1/\varepsilon)$. Larsen and Nelson \cite{LNlinear} showed that this loss can
be removed for linear embeddings on a different family of examples. More precisely, for every $q\geq2$ and
$0<\varepsilon<1/2$, they constructed a $q^{O(1)}$-point subset of
$\R^q$ such that every linear embedding with distortion at most
$1+\varepsilon$ requires
$
 \Omega\left(
 \min\left\{
 q,\varepsilon^{-2}\log q
 \right\}
 \right)
$
dimensions. Larsen and Nelson \cite{LNnonlinear} later obtained lower bounds for arbitrary,
possibly nonlinear, embeddings. They proved that,
whenever
\[
 \varepsilon>
 \frac{\log^{0.5001}n}{\sqrt{\min\{n,d\}}},
\]
there is an $n$-point subset of $\R^d$ for which embedding with distortion at most $1+\varepsilon$
requires
\[
 \Omega\left(
 \varepsilon^{-2}\log(\varepsilon^2n)
 \right)
\]
dimensions. They further conjectured \cite[Conjecture~1]{LNnonlinear} that, throughout the full range of
parameters, the smallest possible target dimension is

\[
 \Theta\left(
 \min\left\{
 d,n-1,\frac{\log(2+\varepsilon^2n)}{\varepsilon^2}
 \right\}
 \right).
\]
The lower-bound direction of this conjecture was subsequently established
by Alon and Klartag \cite{AK}. In the present paper, we prove the
upper bound.

\medskip

Our main result is the following.

\begin{theorem}
\label{thm:jl-main}
There is a universal constant $C>0$ such that the following holds. Let
$n,d\geq2$, let $0<\varepsilon<1/2$, and let $X\subseteq\R^d$ be a set
of $n$ points. Then there is an integer
\[
 r\leq
 C\min\left\{
 d,n-1,\frac{\log(2+\varepsilon^2n)}{\varepsilon^2}
 \right\}
\]
and a linear map $L:(\R^d, \|\cdot \|_2)\to (\R^r, \|\cdot \|_2)$ whose restriction to $X$ has distortion
at most $1+\varepsilon$. 
\end{theorem}

\paragraph{\bf Remark.}
{(i)} The conjecture of Larsen and Nelson allows an arbitrary, possibly
nonlinear, map defined only on $X$. In contrast,
\Cref{thm:jl-main} gives an embedding induced by a linear operator on
the ambient space. Moreover, the matrix walk in the proof can be discretized, along the lines of the Lovett--Meka edge-walk \cite{LM}, to give an efficient randomized procedure for constructing the
operator.

\smallskip

{(ii)} For brevity, write
\begin{equation}\label{eq:dimension-function}
 \Dfun{n}{\varepsilon}
 =
 \varepsilon^{-2}\log(2+\varepsilon^2n).
\end{equation}
The improvement over the previously known upper bound can be substantial:
if $d\geq n$ and $\varepsilon^2n\asymp\log n$, then
\[
 \min\{n,\varepsilon^{-2}\log n\}\asymp n,
 \qquad
 \Dfun{n}{\varepsilon}
 \asymp
 \frac{n\log\log n}{\log n}.
\]

\smallskip

{(iii)} Alon and Klartag~\cite[Theorem~1.2]{AK} proved a
bipartite analogue of \cref{thm:jl-main}. Given vectors
\(a_1,\ldots,a_n,b_1,\ldots,b_n\) in the Euclidean unit ball of $\R^n$, they
construct vectors
\[
 x_1,\ldots,x_n,y_1,\ldots,y_n
 \in \R^{O(\Dfun{n}{\varepsilon})}
\]
such that
\[
 \abs{\ip{x_i}{y_j}-\ip{a_i}{b_j}}
 \leq \varepsilon
 \qquad (i,j\in[n]).
\]
They further conjectured \cite[Conjecture~2.4]{AK} that the output vectors $x_1,\dots,x_n, y_1,\dots,y_n$ in their theorem may be chosen to have norm bounded by a universal constant and proved this conjecture in the regime $\varepsilon = \Theta(1/\sqrt{n})$ \cite[Theorem~2.5]{AK}. We note that our main result \Cref{thm:jl-main}, applied to $\{0,a_1,\dots,a_n, b_1,\dots, b_n\}$, confirms their conjecture.  

In terms of techniques, the bipartite theorem~\cite[Theorem~1.2]{AK} is obtained by using the Khatri--Šidák inequality, Gaussian isoperimetry, and the low-$M^*$ estimate. More relevant to our proof, \cite[Theorem~2.5]{AK} is a dimension-halving statement, proved using
Khatri--Šidák together with the finite-volume ratio theorem. However,
as discussed in~\cite[Section~7]{AK}, this dimension-halving step
cannot be iterated to prove~\cite[Conjecture~2.4]{AK}: the norms of the output vectors $x_1,\dots,x_n, y_1,\dots, y_n$ may increase by a fixed constant, so repeated applications can cause the accumulated error to grow too
quickly. Our dimension-halving step (\Cref{thm:psd}), whose proof is inspired by techniques in algorithmic discrepancy, avoids this deterioration and can therefore be repeatedly applied, with the errors from successive halvings forming a geometrically summable sequence.

\subsection{Rank reduction}
A useful equivalent formulation of linear dimension reduction is as a
positive-semidefinite rank-reduction problem. This formulation appears
explicitly, for instance, in work of So, Ye, and Zhang \cite{SYZ}, who
treat the Johnson--Lindenstrauss problem as a special case of approximate
semidefinite rank reduction.

Indeed, let $u_1,\ldots,u_m\in\R^k$ be unit vectors. If
$L:\R^k\to\R^s$ is linear and
\[
 M=L^{\mathsf T}L,
\]
then $M\succeq0$, $\rank M\leq s$, and
\[
 \norm{Lu_\ell}^2=u_\ell^{\mathsf T}Mu_\ell
\]
for every $\ell\leq m$. Conversely, every positive-semidefinite matrix
$M$ of rank at most $s$ can be written in this form. Thus, on taking
the $u_\ell$ to be the normalized nonzero pairwise differences of a
point configuration, linear dimension reduction amounts to finding a
low-rank positive-semidefinite matrix satisfying
\[
 u_\ell^{\mathsf T}Mu_\ell\approx1
 \qquad (\ell\leq m).
\]

Our main new ingredient is the following rank-halving step.

\begin{theorem}
\label{thm:psd}
Let $r,m\geq1$ and let $u_1,\ldots,u_m\in\R^{2r}$ be unit vectors. There is
a matrix $M\succeq0$ such that
\[
 \rank M\leq r
\]
and
\[
 \max_{\ell\leq m}
 \abs{u_\ell^{\mathsf T}Mu_\ell-1}
 \leq
 C\sqrt{\frac{\log(2+m/r^2)}{r}},
\]
where $C>0$ is a universal constant.
\end{theorem}

\paragraph{\bf Remark.}
For comparison, let $G$ be an $r\times 2r$ matrix with independent
$N(0,1/r)$ entries and set $M=G^{\mathsf T}G$. For every fixed unit
vector $u\in\R^{2r}$, the random variable $\norm{Gu}^2$ has distribution
$\chi_r^2/r$, and hence
\[
 \Prob\left\{
 \abs{\norm{Gu}^2-1}>t
 \right\}
 \leq
 2\exp(-crt^2),
 \qquad 0<t<1,
\]
for a universal constant $c>0$. Thus, in order to control $m$
prescribed unit vectors simultaneously by a union bound, one takes
\[
 t\asymp
 \sqrt{\frac{\log(2+m)}{r}}.
\]
In the regime $r\gtrsim\log(2+m)$, this gives a rank-$r$ matrix with error
\[
 O\left(\sqrt{\frac{\log(2+m)}{r}}\right).
\]
The gain in \Cref{thm:psd} is the replacement of $m$ by $m/r^2$ inside
the logarithm. In particular, for the $m\asymp n^2$ pairwise differences
of an $n$-point configuration, the logarithmic term is of order
$\log(2+n/r)$ rather than $\log n$. 

\medskip

The construction of $M$ uses ideas from algorithmic discrepancy theory. Briefly, we perform an adaptive random walk in the space of symmetric matrices
while controlling the prescribed quadratic forms. We describe the
construction and the subsequent iteration in the proof overview (see \Cref{sec:overview}).

\subsection{Sketching inner products}

Let $X=(x_1,\ldots,x_n)$ be an indexed configuration in the Euclidean
unit ball of $\R^k$.  Following Alon and Klartag \cite{AK}, an
$\varepsilon$-inner-product sketch for $X$ is a data structure from which, given
$i,j\in[n]$, one can recover $\ip{x_i}{x_j}$ up to additive error $\varepsilon$.  Let
$f(n,k,\varepsilon)$ denote the minimum number of bits required in the
worst case.  

For
\[
 \varepsilon\geq {\frac{2}{\sqrt{n}}},
 \qquad
 t=\Dfun{n}{\varepsilon}
 =\frac{\log(2+\varepsilon^2n)}{\varepsilon^2},
\]
Alon and Klartag proved \cite[Theorem~2.2]{AK} that
\[
 \Omega(nt)
 \leq
 f(n,k,\varepsilon)
 \leq
 O\left(\frac{n\log n}{\varepsilon^2}\right)
 \qquad
 \text{for } t\leq k\leq n.
\]
In the remaining regimes for $k$, they obtained matching upper and lower bounds. In particular, at the threshold dimension $k= \lceil t \rceil$,
their result gives
\[
 f(n, \lceil t \rceil,\varepsilon)=\Theta(nt).
\]
They conjectured that the lower bound $\Omega(nt)$ is also sharp throughout the regime $k\geq t$.

It is useful to explain why this does not follow immediately from
dimension reduction. Suppose that one first embeds the configuration
into the Euclidean unit ball of $\R^s$ and then rounds each vector to an
$\varepsilon$-net. Since such a net may be chosen to have cardinality at
most
$
 \left(\frac{C}{\varepsilon}\right)^s,
$
this gives a sketch of length
$
 O\left(ns\log(1/\varepsilon)\right).
$
Even if one could first reduce to the target dimension $s\asymp t$, this
procedure would therefore use
\[
 O\left(nt\log(1/\varepsilon)\right)
\]
bits, rather than the conjectured order $nt$. 

Instead, Alon and Klartag's upper bounds in the lower-dimensional regimes for $k$ are
based on an $\ell_\infty$-covering estimate for Gram matrices. In particular, allowing fixed constant factors in the accuracy and
dimension, their estimates give an $(\varepsilon/2)$-inner-product
sketch of length $O(nt)$ for configurations in dimension $O(t)$.
The remaining task in the high-dimensional regime is therefore to
reduce an arbitrary configuration to this dimension while approximately
preserving inner products. Our main theorem provides precisely this
reduction and consequently gives the following.

\begin{theorem}
\label{thm:sketch-main}
Let $n\geq2$ and suppose that
\[
 \frac{2}{\sqrt n}\leq\varepsilon\leq\frac12,
 \qquad
 \Dfun{n}{\varepsilon}\leq k\leq n.
\]
Then
\[
 f(n,k,\varepsilon)
 =
 \Theta\left(
 n\Dfun{n}{\varepsilon}
 \right).
\]
\end{theorem}

\paragraph{\bf Remark.}
Together with \cite[Theorem~2.2]{AK},
\Cref{thm:sketch-main} determines $f(n,k,\varepsilon)$ up to universal
constant factors throughout the nontrivial parameter range $1\leq k \leq n$ and $2/\sqrt{n} \leq \varepsilon \leq 1/2$. 

\begin{proof}[Proof of \Cref{thm:sketch-main}]
The lower bound is \cite[Theorem~2.2]{AK}. For the upper bound, apply
\Cref{thm:jl-main} with distortion parameter $\varepsilon/12$ to
$\{0,x_1,\ldots,x_n\}$. After rescaling the resulting linear map, we
obtain $L:\R^k\to\R^s$, where
\[
 s=O(\Dfun{n}{\varepsilon}),
\]
such that
\[
 (1+\varepsilon/12)^{-1}\norm{u-v}
 \leq \norm{Lu-Lv}
 \leq \norm{u-v}
\]
for all $u,v\in\{0,x_1,\ldots,x_n\}$. Set $z_i=Lx_i$. The vectors $z_i$
lie in the Euclidean unit ball, and polarization gives
\[
 2\abs{\ip{z_i}{z_j}-\ip{x_i}{x_j}}
 \leq
 \left(1-(1+\varepsilon/12)^{-2}\right)
 \left(
 \norm{x_i}^2+\norm{x_j}^2+\norm{x_i-x_j}^2
 \right) \leq\varepsilon.
\]
Thus the inner products of the $z_i$ approximate those of the $x_i$ to
within $\varepsilon/2$. By the lower-dimensional estimate in \cite[Theorem~2.2]{AK} described
above, $(z_1,\ldots,z_n)$ has an $(\varepsilon/2)$-inner-product sketch
using $O(n\Dfun{n}{\varepsilon})$ bits. The same sketch therefore
approximates the inner products of the original configuration to within
$\varepsilon$.
\end{proof}

\subsection{Proof overview}
\label{sec:overview}

The main theorem follows from \Cref{thm:psd} by iterating the
rank-halving step, so we focus on the proof of \Cref{thm:psd}.

We begin from $M=I_{2r}$ and repeatedly modify $M$ so as to decrease its
rank while changing each prescribed quadratic form only by a small
relative amount. Suppose that at some stage $M\succeq0$ has rank $k>r$.
We seek a new positive-semidefinite matrix $M'$ whose rank is smaller
than that of $M$, but for which
$
 u_\ell^{\mathsf T}M'u_\ell
$
remains close to $u_\ell^{\mathsf T}Mu_\ell$ for every $\ell\leq m$.

Choose a symmetric operator $Q$ on the $k$-dimensional space
$\operatorname{range}(M)$ and consider the update
\[
 M'
 =
 M^{1/2}(I+\sigma Q)M^{1/2},
 \qquad \sigma\in\{-1,1\}.
\]
For $M'$ to remain positive semidefinite, it is
enough to require
\[
 -I\preceq Q\preceq I.
\]
To decrease the rank by a fixed proportion, we would like
$I+\sigma Q$ to have a large kernel. Thus it is enough for $Q$ to have
many eigenvalues equal to $1$ or $-1$, since one choice of $\sigma$
then turns a fixed proportion of them into zero eigenvalues of
$I+\sigma Q$. It remains to understand what this update does to the prescribed
quadratic forms. Setting
\[
 w_\ell
 =
 \frac{M^{1/2}u_\ell}
 {\sqrt{u_\ell^{\mathsf T}Mu_\ell}},
\]
we obtain
\[
 u_\ell^{\mathsf T}M'u_\ell
 =
 u_\ell^{\mathsf T}Mu_\ell
 \left(1+\sigma w_\ell^{\mathsf T}Qw_\ell\right).
\]
Thus relative preservation of all the quadratic forms is guaranteed if
the quantities $w_\ell^{\mathsf T}Qw_\ell$ are uniformly small. The
rank-halving step is therefore reduced to constructing a symmetric
contraction $Q$ with many eigenvalues equal to $\pm1$ and with small
quadratic forms on a prescribed family of unit vectors.

We construct $Q$ using a matrix version of the partial-coloring
argument of Lovett and Meka \cite{LM} from their constructive proof of
Spencer's theorem in discrepancy theory. Starting from $Q_0=0$, we run
Brownian motion in the space of symmetric matrices. Once an eigenvalue
reaches $1$ or $-1$, future increments are chosen to leave the
corresponding eigenspace fixed. The idea of freezing eigendirections at the boundary and continuing on the remaining subspace is classical in semidefinite rank reduction \cite{Barvinok,Pataki}. For matrices satisfying $0\preceq X\preceq I$, an analogous two-sided construction is used in the SDP iterative-rounding framework of Tantipongpipat et al.~\cite{TSSMV19}, where both the $0$- and $1$-eigenspaces are frozen. In our construction, once $\lvert w_\ell^{\mathsf T}Q_tw_\ell\rvert$ reaches a prescribed level $b$, future increments are also required to leave this quadratic form fixed.

Before $k/4$ eigenvalues have reached $\pm1$, the remaining eigenspace
has dimension at least $3k/4$, so the symmetric matrices supported there
form a space of dimension $\Theta(k^2)$. Each fixed quadratic form
imposes only one additional linear constraint. Hence, until $k^2/16$
quadratic forms have been fixed, the process still moves in
$\Omega(k^2)$ directions and its expected
squared Frobenius norm accumulates at rate $\Omega(k^2)$. Since $Q_t$ is a
contraction, we always have $\|Q_t\|_F^2\leq k$. Thus, if $p$ denotes
the probability that neither threshold is reached by time $T$, then the
It\^o isometry gives
\[
 k
 \geq
 \E\|Q_T\|_F^2
 \gtrsim
 p\,k^2T.
\]
Taking $T=C/k$ with $C$ sufficiently large therefore makes $p$ a small
constant.

On the other hand, each process $w_\ell^{\mathsf T}Q_tw_\ell$ is a
one-dimensional martingale whose quadratic variation grows at rate at
most one. Hence, over the time interval $[0,T]$, where $T=C/k$, the
probability that it reaches magnitude $b$ is at most
\[
 \exp(-c k b^2).
\]
There are $m$ quadratic forms, so the expected number that reach the
barrier by time $T$ is at most
\[
 m\exp(-c k b^2).
\]
To ensure that fewer than $k^2/16$ quadratic forms are fixed with
constant probability, it therefore suffices, by Markov's inequality, to
choose $b$ so that
\[
 m\exp(-c k b^2)\lesssim k^2.
\]
This gives
\[
 b
 \asymp
 \sqrt{\frac{\log(2+m/k^2)}{k}},
\]
which is the key quantitative gain over the bound
$\sqrt{\log(2+m)/k}$ coming from a random Gaussian matrix. Indeed, for
an $n$-point configuration we have $m\asymp n^2$, so a rank-halving
step from dimension $2r$ to $r$ incurs error
\[
 O\left(
 \sqrt{\frac{\log(2+n/r)}{r}}
 \right).
\]
Iterating down to a target dimension $r_0$ gives a geometrically
summable sequence of errors, with total error bounded by
\[
 O\left(
 \sqrt{\frac{\log(2+n/r_0)}{r_0}}
 \right).
\]
Taking $r_0\asymp\Dfun{n}{\varepsilon}$ makes this $O(\varepsilon)$
and yields \Cref{thm:jl-main}.

\subsection{Notation}
All logarithms are natural.  For a positive integer $N$, write
$[N]=\{1,\ldots,N\}$.  For $x,y\in\R^k$, $\norm{x}$ and $\ip{x}{y}$ denote the Euclidean norm
and inner product, and
$S^{k-1}=\{x\in\R^k:\norm{x}=1\}$.  Let $\Sym_k$ be the space of real
symmetric $k\times k$ matrices.  For $A,B\in\Sym_k$, the notation
$A\preceq B$ means that $B-A$ is positive semidefinite; in particular,
$A\succeq0$ means that $A$ is positive semidefinite.  We use
\[
 \langle A,B\rangle_F=\tr(AB),
 \qquad
 \|A\|_F=\langle A,A\rangle_F^{1/2}
\]
for the Frobenius inner product and norm.  Eigenvalues are counted with
multiplicity, and $P_F$ denotes the orthogonal projection onto $F$.
For an operator on a
finite-dimensional inner-product space $H$, $\|\cdot\|_{\mathrm{HS}(H)}$
denotes the Hilbert--Schmidt norm.  For a continuous local martingale
$Z$, $[Z]_t$ denotes its quadratic variation, and $\chi_r^2$ denotes the
chi-square distribution with $r$ degrees of freedom.

We will also make use of asymptotic notation. For functions $f,g$, $f = O_{\alpha}(g)$ (or $f\lesssim_{\alpha} g$) means that $f \le C_\alpha g$, where $C_\alpha$ is some constant depending on $\alpha$; $f = \Omega_{\alpha}(g)$ (or $f \gtrsim_{\alpha} g$) means that $f \ge c_{\alpha} g$, where $c_\alpha > 0$ is some constant depending on $\alpha$, and $f = \Theta_{\alpha}(g)$ means that both $f = O_{\alpha}(g)$ and $f = \Omega_{\alpha}(g)$ hold. 

\subsection*{Acknowledgments and AI disclosure}
Motivated by \cite[Theorem~2.5]{AK}, the author prompted GPT 5.4 Pro to prove a version of rank-halving which can be iterated to prove the theorem. Although this attempt was unsuccessful, sustained interaction led to the author identifying the PSD rank reduction formulation in \cite{SYZ} and the broader literature on SDP rank reduction as promising candidates. Subsequently, the author prompted GPT 5.6 Sol with the PSD formulation and references to SDP rank reduction, with a directive to find a suitable construction. This led to the model returning the Lovett-Meka inspired construction, which is the basis of the proof in \cref{sec:spectral-proof}. In \cref{sec:overview}, we have noted the similarity of this construction to previous constructions in the SDP rank reduction literature, perhaps optically most closely to the work of Tantipongpipat et al.~\cite{TSSMV19}.  

After the appearance of the first version of this paper on arXiv, we were informed by Or Zamir of a separate ChatGPT interaction, predating our posting, in which GPT-5.6 Sol produced an essentially identical proof. Zamir describes his mathematical input as mostly limited to identifying the desired rank halving objective. In Zamir's interaction, the model itself introduced the PSD formulation and, after generic steering that included no substantive mathematical ideas, found a proof using Gaussian correlation and Gaussian isoperimetry, using the partial coloring framework of Rothvoss \cite{rothvoss2017constructive}. Instructions by Zamir to find a more combinatorial argument led to the model returning the same Lovett-Meka inspired matrix walk. 
 
The author used Codex for assistance with preparing the manuscript. The mathematical content, the final text, and any  errors are the responsibility of the author.  The author is partially supported by NSF grant DMS-2237646.

\section{Dimension reduction by iterated rank halving}
\label{sec:reduction}

In this section, we deduce \Cref{thm:jl-main} from the rank-halving
theorem (\Cref{thm:psd}) stated in the introduction. We will use the following simple corollary of \Cref{thm:psd}. 

\begin{corollary}\label{cor:one-step}
There is a universal constant $C_0>0$ such that the following holds.
Let $n\geq2$ and $r\geq1$. For every finite set
$X\subseteq\R^{2r}$ with $2\leq\abs{X}\leq n$, there is a linear map
$L:\R^{2r}\to\R^r$ such that
\[
 (1-\delta_r)\norm{x-y}^2
 \leq
 \norm{Lx-Ly}^2
 \leq
 (1+\delta_r)\norm{x-y}^2
\]
for all $x,y\in X$, where
\[
 \delta_r
 =
 C_0\sqrt{\frac{\log(2+n/r)}{r}}.
\]
\end{corollary}

\begin{proof}
Write $X=\{x_1,\ldots,x_N\}$, where $N\leq n$, and apply
\Cref{thm:psd} to the unit vectors
\[
 u_{ij}
 =
 \frac{x_i-x_j}{\norm{x_i-x_j}},
 \qquad 1\leq i<j\leq N.
\]
There are $\binom N2\leq\binom n2$ such vectors. Factoring the resulting
matrix as $M=L^{\mathsf T}L$ gives the desired linear map, since
\[
 \frac{\norm{Lx_i-Lx_j}^2}{\norm{x_i-x_j}^2}
 =
 u_{ij}^{\mathsf T}Mu_{ij}.
\]
Finally,
\[
 \log\left(2+\frac{\binom n2}{r^2}\right)
 \leq
 2\log\left(2+\frac nr\right),
\]
so the constant in \Cref{thm:psd} may be absorbed into a universal
constant $C_0$.
\end{proof}

We now present the details of the iteration. 

\begin{proof}[Proof of \Cref{thm:jl-main} given \Cref{thm:psd}]
Translate the configuration by $-x_1$ and write
\[
 \widetilde X=\{x-x_1:x\in X\},
 \qquad
 V=\spann(\widetilde X).
\]
Translation does not change pairwise distances, and
\[
 \dim V\leq\min\{d,n-1\}.
\]
It suffices to construct the desired linear map on $V$, since the map
may then be extended by zero on $V^\perp$.

Let $C_0$ be the constant in \Cref{cor:one-step}, and choose a
sufficiently large universal constant $K$. Set
\[
 r_0=
 \left\lceil
 K\Dfun{n}{\varepsilon}
 \right\rceil.
\]
If $\dim V\leq r_0$, an isometric identification of $V$ with
$\R^{\dim V}$ gives an exact realization. Its target dimension satisfies
\[
 \dim V\leq\min\{d,n-1,r_0\}
 \leq
 C\min\{d,n-1,\Dfun{n}{\varepsilon}\}.
\]
We may therefore assume that $\dim V>r_0$.

Choose $s\geq1$ minimal such that
\[
 \dim V\leq2^sr_0.
\]
For $0\leq j\leq s-1$, set
\[
 \beta_j
 =
 C_0\sqrt{
 \frac{\log(2+n/(2^jr_0))}{2^jr_0}
 },
 \qquad
 S=\sum_{j=0}^{s-1}\beta_j.
\]
Since $j\mapsto\log(2+n/(2^jr_0))$ is decreasing,
\[
 S
 \leq
 C_0\sum_{j=0}^{\infty}2^{-j/2}
 \sqrt{\frac{\log(2+n/r_0)}{r_0}}
 \leq \frac{C\varepsilon}{\sqrt K}.
\]
Indeed, for $K$ larger than a universal constant,
\[
 r_0
 \geq K\varepsilon^{-2}\log(2+\varepsilon^2n)
 \geq \varepsilon^{-2},
\]
and hence
\[
 \log(2+n/r_0)\leq\log(2+\varepsilon^2n).
\]
Choose $K$ sufficiently large that $S\leq\varepsilon/10$. In
particular, every $\beta_j$ lies in $[0,1)$.

Fix an isometric linear embedding
\[
 \iota:V\longrightarrow\R^{2^sr_0}
\]
and set $X_s=\iota(\widetilde X)$. For
$j=s-1,s-2,\ldots,0$, apply \Cref{cor:one-step} to $X_{j+1}$, with
target dimension $2^jr_0$, obtaining a linear map
\[
 L_j:\R^{2^{j+1}r_0}\longrightarrow\R^{2^jr_0},
 \qquad
 X_j=L_j(X_{j+1}).
\]
The squared-distance error at this step is at most $\beta_j$. Note that since
$\beta_j<1$, the map $L_j$ is injective on $X_{j+1}$.

Let
\[
 A=L_0\circ L_1\circ\cdots\circ L_{s-1}\circ\iota.
\]
For every distinct $x,y\in X$, the ratio between the final and initial
squared distances lies between
\[
 \prod_{j=0}^{s-1}(1-\beta_j)
 \quad\text{and}\quad
 \prod_{j=0}^{s-1}(1+\beta_j).
\]
Consequently,
\[
 \prod_{j=0}^{s-1}(1-\beta_j)\geq1-S,
 \qquad
 \prod_{j=0}^{s-1}(1+\beta_j)\leq e^S.
\]
The distortion of $A$ on $\widetilde X$ is therefore at most
\[
 \sqrt{\frac{e^S}{1-S}}
 \leq e^{3S/2}
 \leq1+\varepsilon,
\]
where we used $S\leq\varepsilon/10$ and
$0<\varepsilon<1/2$.

The target dimension of this embedding is $r_0$.  Since $r_0<\dim V$ in the case under
consideration and $r_0\leq C\Dfun{n}{\varepsilon}$,
\[
 r_0\leq
 C\min\left\{
 d,n-1,\Dfun{n}{\varepsilon}
 \right\}
\]
for a universal constant $C$.  Finally, extend $A$ by zero on
$V^\perp$.  For $x,y\in X$, this extension sends $x-y$ to
\[
 A(x-y)=A\bigl((x-x_1)-(y-x_1)\bigr),
\]
so it has the same distortion on the original configuration $X$.
\end{proof}

\section{Proof of the rank-halving theorem}
\label{sec:rank-halving}

The main ingredient in the proof of \Cref{thm:psd} is the following
proposition, which produces a symmetric contraction with many eigenvalues
equal to $1$ or $-1$, while keeping a prescribed family of quadratic
forms small. Its proof is deferred to \Cref{sec:spectral-proof}.

\begin{proposition}\label{prop:spectral}
There is a universal constant $C_1>0$ such that the following holds.
Let $k,m\geq1$ and let $w_1,\ldots,w_m\in S^{k-1}$. Then there exists
$Q\in\Sym_k$ such that
\[
 -I\preceq Q\preceq I,
\]
at least $\lceil k/4\rceil$ eigenvalues of $Q$ are equal to $1$ or
$-1$, and
\[
 \max_{\ell\leq m}
 \abs{w_\ell^{\mathsf T}Qw_\ell}
 \leq
 C_1\sqrt{\frac{\log(3+m/k^2)}{k}}.
\]
\end{proposition}

\begin{proof}[Proof of \Cref{thm:psd} given \Cref{prop:spectral}]
Let
\[
 \eta
 =
 C_1\sqrt{\frac{\log(3+m/r^2)}{r}},
\]
where $C_1$ is the constant in \Cref{prop:spectral}.  Since
$\log(3+x)\asymp\log(2+x)$ uniformly for $x\geq0$, if
$\eta\geq1/10$, then $M=0$ gives the result after increasing the
universal constant in \Cref{thm:psd}. We may therefore assume that
$\eta<1/10$.

We construct positive-semidefinite matrices $M_0,M_1,\ldots$, starting
from $M_0=I_{2r}$. At each stage, the new rank is at most seven eighths
of the preceding rank, while every prescribed
quadratic form changes by a multiplicative factor in
$[1-\eta,1+\eta]$.

Suppose that $M_j$ has rank $k_j>r$. Since all prescribed quadratic
forms are initially equal to $1$ and remain positive throughout the
construction, we may define
\[
 w_{\ell,j}
 =
 \frac{M_j^{1/2}u_\ell}
 {\sqrt{u_\ell^{\mathsf T}M_ju_\ell}},
 \qquad \ell\leq m.
\]
These are unit vectors in
$\operatorname{range}(M_j)\simeq\R^{k_j}$.

Apply \Cref{prop:spectral} to $w_{1,j},\ldots,w_{m,j}$, and let $Q_j$
be the resulting symmetric operator on $\operatorname{range}(M_j)$,
extended by zero on $\ker M_j$. Since at least $k_j/4$ eigenvalues of
$Q_j$ are equal to $1$ or $-1$, one of the two signs has multiplicity
at least $k_j/8$. Choose $\sigma_j\in\{-1,1\}$ so that
\[
 \dim\ker\left(
 (I+\sigma_jQ_j)|_{\operatorname{range}(M_j)}
 \right)
 \geq\frac{k_j}{8},
\]
and set
\[
 M_{j+1}
 =
 M_j^{1/2}(I+\sigma_jQ_j)M_j^{1/2}.
\]

Since $-I\preceq Q_j\preceq I$, the matrix $M_{j+1}$ is positive
semidefinite, and
\[
 \rank M_{j+1}\leq\frac{7k_j}{8}.
\]
For every $\ell\leq m$,
\[
 u_\ell^{\mathsf T}M_{j+1}u_\ell
 =
 u_\ell^{\mathsf T}M_ju_\ell
 \left(1+\theta_{\ell,j}\right),
\]
where
\[
 \theta_{\ell,j}
 =
 \sigma_j w_{\ell,j}^{\mathsf T}Q_jw_{\ell,j}.
\]
The function
$
 x\mapsto{\log(3+m/x^2)}/{x}
$
is decreasing on $(0,\infty)$.  Since $k_j>r$, it follows that
\[
 \abs{\theta_{\ell,j}}
 \leq
 C_1\sqrt{\frac{\log(3+m/k_j^2)}{k_j}}
 \leq\eta.
\]
Thus every factor $1+\theta_{\ell,j}$ is positive, so the iteration may
continue.  Since
$
 2\left(\frac78\right)^6<1,
$
after at most six steps we obtain a matrix $M_s$ with
$\rank M_s\leq r$. For every $\ell\leq m$,
\[
 u_\ell^{\mathsf T}M_su_\ell
 =
 \prod_{j=0}^{s-1}(1+\theta_{\ell,j}),
 \qquad s\leq6.
\]
Since $\eta<1/10$,
\[
 \abs{u_\ell^{\mathsf T}M_su_\ell-1}
 \leq C\eta
 \leq C\sqrt{\frac{\log(2+m/r^2)}{r}},
\]
where the last comparison again uses
$\log(3+x)\asymp\log(2+x)$. This completes the proof.
\end{proof}

\section{Construction of the symmetric contraction}
\label{sec:spectral-proof}

Finally, we prove \Cref{prop:spectral}.

\begin{proof}[Proof of \Cref{prop:spectral}]
Choose the universal constant $C_1\geq 40$, and set
\[
 \Lambda=\log\left(3+\frac{m}{k^2}\right),
 \qquad
 b=C_1\sqrt{\frac{\Lambda}{k}},
 \qquad
 T=\frac{C_1}{k}.
\]
If $b\geq1$, we may take $Q=I$, so assume that $b<1$.

Let $(B_t)_{t\geq0}$ be standard Brownian motion in the Euclidean space
$(\Sym_k,\langle\cdot,\cdot\rangle_F)$, and let
$(\mathcal F_t)_{t\geq0}$ be its natural filtration. We iteratively
construct a continuous adapted process $(Q_t)_{t\geq0}$, beginning with
\[
 \tau_0=0,
 \qquad
 Q_{\tau_0}=0.
\]
The construction will terminate when $t=T$ or when either of the
following conditions is reached:
\begin{equation}
\label{eq:walk-counting-thresholds}
 \dim\bigl(\ker(I-Q_t)\oplus\ker(I+Q_t)\bigr)
 \geq\left\lceil\frac{k}{4}\right\rceil
 \quad\text{or}\quad
 \#\left\{
 \ell\leq m:
 \abs{w_\ell^{\mathsf T}Q_tw_\ell}=b
 \right\}
 \geq\left\lceil\frac{k^2}{16}\right\rceil .
\end{equation}

\paragraph{\bf Construction of the process.}
Suppose that $Q_t$ has been defined up to a stopping time
$\tau_j<T$ and that neither condition in
\eqref{eq:walk-counting-thresholds} holds at time $\tau_j$. Define
\[
 E_{+,j}=\ker(I-Q_{\tau_j}),
 \qquad
 E_{-,j}=\ker(I+Q_{\tau_j}),
 \qquad
 F_j=(E_{+,j}\oplus E_{-,j})^\perp,
\]
and
\[
 J_j=
 \left\{
 \ell\leq m:
 \abs{w_\ell^{\mathsf T}Q_{\tau_j}w_\ell}=b
 \right\}.
\]
Let
\[
 \mathcal V_j=
 \left\{
 H\in\Sym_k:
 H=P_{F_j}HP_{F_j},
 \quad
 w_\ell^{\mathsf T}Hw_\ell=0
 \text{ for every }\ell\in J_j
 \right\},
\]
and let $\Pi^{(j)}$ be the Frobenius-orthogonal projection of $\Sym_k$
onto $\mathcal V_j$. All the objects defining $\Pi^{(j)}$ are
$\mathcal F_{\tau_j}$-measurable.

Starting from time $\tau_j$, evolve the process according to
\[
 Q_t
 =
 Q_{\tau_j}+\Pi^{(j)}(B_t-B_{\tau_j})
\]
until the first time $t>\tau_j$ at which either the restriction
$Q_t|_{F_j}$ has an eigenvalue equal to $1$ or $-1$, or
$
 \abs{w_\ell^{\mathsf T}Q_tw_\ell}=b
$
for some $\ell\notin J_j$. Let $\tau_{j+1}$ be the lesser of $T$ and
this first hitting time. The tracked quantities are continuous and
adapted, so $\tau_{j+1}$ is a stopping time; since they lie strictly
inside their respective barriers at time $\tau_j$, we also have
$\tau_{j+1}>\tau_j$.

By the strong Markov property,
$
 \bigl(B_{\tau_j+s}-B_{\tau_j}\bigr)_{s\geq0}
$
is standard Brownian motion in $\Sym_k$, independent of
$\mathcal F_{\tau_j}$. Since $\Pi^{(j)}$ is
$\mathcal F_{\tau_j}$-measurable,
$
 \left(\Pi^{(j)}(B_{\tau_j+s}-B_{\tau_j})\right)_{s\geq0}
$
is, conditioned on $\mathcal F_{\tau_j}$, standard Brownian motion in
$\mathcal V_j$.

The condition $H=P_{F_j}HP_{F_j}$ ensures that the eigenspaces
$E_{+,j}$ and $E_{-,j}$ remain fixed at eigenvalues $1$ and $-1$,
respectively, throughout this stage. Likewise, the conditions indexed
by $J_j$ ensure that every quadratic form which has reached a barrier
remains fixed there. Consequently,
\[
 E_{+,j}\subseteq E_{+,j+1},
 \qquad
 E_{-,j}\subseteq E_{-,j+1},
 \qquad
 J_j\subseteq J_{j+1},
\]
and these eigenspaces and quadratic forms remain fixed at every later
stage.

If $\tau_{j+1}=T$ or either condition in
\eqref{eq:walk-counting-thresholds} holds at time $\tau_{j+1}$,
terminate the construction; otherwise proceed to the next stage.
Every stage ending before $T$ adds at least one new eigendirection to
$E_{+,j}\oplus E_{-,j}$ or at least one new index to $J_j$. Hence only
finitely many stages occur. Denote the terminal time by $\tau$, and set
\[
 Q_t=Q_\tau,
 \qquad t>\tau.
\]

For each stage that is run, define
\[
 \Pi_t=\Pi^{(j)}
 \quad\text{for }\tau_j<t\leq\tau_{j+1},
 \qquad
 \Pi_t=0
 \quad\text{for }t>\tau,
\]
and set $\Pi_0=0$. Since $\Pi^{(j)}$ is
$\mathcal F_{\tau_j}$-measurable, $(\Pi_t)$ is adapted; the endpoint
convention makes it left-continuous. It is therefore predictable, and
\[
 Q_t=\int_0^t\Pi_s\,dB_s
 \qquad (t\geq0).
\]
Moreover, 
\[
 -I\preceq Q_t\preceq I,
 \qquad
 \abs{w_\ell^{\mathsf T}Q_tw_\ell}\leq b
 \quad\text{for every }\ell\leq m
\]
throughout the construction.

\paragraph{\bf Dimension of the available directions.}
Whenever the construction is running, neither condition in
\eqref{eq:walk-counting-thresholds} holds. Since the quantities involved
are integers,
\[
 \dim F_j>\frac{3k}{4},
 \qquad
 \abs{J_j}<\frac{k^2}{16}.
\]
The symmetric matrices supported on $F_j$ form a vector space of
dimension
\[
 \frac{\dim F_j(\dim F_j+1)}{2},
\]
and each index in $J_j$ imposes at most one linear condition. Hence
\[
 \dim\mathcal V_j
 \geq
 \frac{\dim F_j(\dim F_j+1)}{2}-\abs{J_j}
 >
 \frac{k^2}{5}.
\]

\paragraph{\bf Measuring progress.}
Let $\mathcal E$ be the event that neither condition in
\eqref{eq:walk-counting-thresholds} is reached by time $T$. On
$\mathcal E$, the construction runs until time $T$, and
\[
 \rank(\Pi_t)>\frac{k^2}{5}
\]
for almost every $t\in[0,T]$.

Since $Q_T=Q_\tau$ is a contraction,
\[
 \|Q_\tau\|_F^2\leq k.
\]
On the other hand, the It\^o isometry and the fact that each $\Pi_t$ is
an orthogonal projection give
\[
 \E\|Q_\tau\|_F^2
 =
 \E\|Q_T\|_F^2
 =
 \E\int_0^T
 \|\Pi_t\|_{\mathrm{HS}(\Sym_k)}^2\,dt
 =
 \E\int_0^\tau\rank(\Pi_t)\,dt.
 \]
Therefore
\[
 k
 \geq
 \E\|Q_\tau\|_F^2
 \geq
 \Prob(\mathcal E)\frac{k^2}{5}T
 =
 \Prob(\mathcal E)\frac{C_1k}{5},
\]
and hence
\[
 \Prob(\mathcal E)
 \leq\frac{5}{C_1}
 \leq\frac18.
\]

\paragraph{\bf Quadratic forms reaching the barriers.}
Fix $\ell\leq m$ and, for $0\leq t\leq T$, define
\[
 Z_\ell(t)
 =
 w_\ell^{\mathsf T}Q_tw_\ell
 =
 \left\langle
 w_\ell w_\ell^{\mathsf T},Q_t
 \right\rangle_F.
\]
Since $\Pi_t$ is a self-adjoint orthogonal projection on $\Sym_k$,
\[
 dZ_\ell(t)
 =
 \left\langle
 \Pi_t(w_\ell w_\ell^{\mathsf T}),dB_t
 \right\rangle_F.
\]
Thus $Z_\ell$ is a continuous martingale with quadratic variation
\[
 d[Z_\ell]_t
 =
 \left\|
 \Pi_t(w_\ell w_\ell^{\mathsf T})
 \right\|_F^2\,dt
 \leq dt,
\]
because $\|w_\ell w_\ell^{\mathsf T}\|_F^2=1$. In particular,
$[Z_\ell]_T\leq T$. Bernstein's inequality for continuous martingales
(see, e.g., \cite[Chapter~IV, Exercise~3.16]{RY}) therefore gives
\[
 \Prob\left\{
 \sup_{0\leq t\leq T}\abs{Z_\ell(t)}\geq b
 \right\}
 \leq
 2\exp\left(-\frac{b^2}{2T}\right)
 =
 2e^{-C_1\Lambda/2}.
\]

Let $N$ be the number of indices $\ell$ whose quadratic forms reach
either barrier by time $\tau$. By linearity of expectation,
\[
 \E N
 \leq
 2me^{-C_1\Lambda/2}
 =
 \frac{2m}{(3+m/k^2)^{C_1/2}}
 \leq
 \frac{2k^2}{3^{C_1/2-1}}
 <
 \frac{k^2}{128}.
 \]
Markov's inequality gives
\[
 \Prob\left\{
 N\geq\left\lceil\frac{k^2}{16}\right\rceil
 \right\}
 \leq\frac18.
\]

If the spectral condition in
\eqref{eq:walk-counting-thresholds} is not reached by the terminal time $T$,
then either $\mathcal E$ occurs or at least
$\lceil k^2/16\rceil$ quadratic forms have reached their barriers.
Consequently, with probability at least $3/4$,
\[
 \dim\bigl(\ker(I-Q_\tau)\oplus\ker(I+Q_\tau)\bigr)
 \geq\left\lceil\frac{k}{4}\right\rceil.
\]
Choose any realization of this event and set $Q=Q_\tau$. Then
\[
 -I\preceq Q\preceq I,
 \qquad
 \#\{i:\abs{\lambda_i(Q)}=1\}
 \geq\left\lceil\frac{k}{4}\right\rceil,
\]
and
\[
 \max_{\ell\leq m}
 \abs{w_\ell^{\mathsf T}Qw_\ell}
 \leq
 b
 =
 C_1\sqrt{\frac{\log(3+m/k^2)}{k}}.
\]
This proves \Cref{prop:spectral}.
\end{proof}

\bibliographystyle{amsplain0}
\bibliography{main}

\end{document}